\documentclass[a4paper,11pt]{article}
\usepackage{amsmath}
\usepackage{amsfonts,amssymb}
\usepackage{eucal}
\usepackage{amsthm}
\usepackage[pdftex]{graphicx}
\usepackage{url}
\numberwithin{equation}{section}

\newcommand{\simgt}{\lower.5ex\hbox{$\; \buildrel > \over \sim \;$}}
\newcommand{\simlt}{\lower.5ex\hbox{$\; \buildrel < \over \sim \;$}}

\renewcommand{\Re}{\text{{\rm Re}}}
\renewcommand{\Im}{\text{{\rm Im}}}

\newcommand{\dist}{\text{\rm dist}}

\newcommand{\sgn}{\text{\rm sgn}}
\newcommand{\loc}{\text{\rm loc}}

\def\qed{\hfill $\Box$}

\newcommand{\tr}{\text{\rm tr}}

\newtheorem{thm}{Theorem}
\newtheorem{lem}{Lemma}
\newtheorem{prop}{Proposition}
\newtheorem{cor}{Corollary}

\theoremstyle{definition}
\newtheorem{defn}{Definition}

\newtheorem{example}{Example}

\theoremstyle{remark}
\newtheorem{rem}{Remark}

\begin{document}
\title{Perturbation of discriminant for one-dimensional  discrete Schr\"odinger operator with sparse periodic potential}
\author{Masahiro Kaminaga}
\date{}
\maketitle

\begin{abstract}
We consider the one-dimensional discrete Schr\"odinger operator with complex-valued sparse periodic potential. 
The spectrum for a complex-valued periodic potential is a complicated compact set in the complex plane represented by real intersections of algebraic curves determined by a discriminant.
We represent the discriminant by Chebyshev polynomials and use perturbations of the discriminant to study the spectrum.
\end{abstract}

\section{Introduction}
In this paper, we consider the one-dimensional discrete Schr\"odinger operator:
$$ 
(Hu)(n) = u(n+1)+u(n-1) + V(n)u(n)\quad\mbox{on $l^2(\mathbb{Z})$}, 
$$ 
where $V(n)$ is a complex-valued periodic potential with period $L$ 
containing only one nonzero value $v$ within a period. 
As the period $L$ increases, this potential becomes sparse, which we call the sparse potential. 
If $v$ is real, $H$ is a bounded self-adjoint operator. 
For real $v$, the spectrum $\sigma(H)$ of $H$ is purely absolutely continuous and consists of at most $L$ closed intervals on the real axis.
These results are proved in general dimension using the method of direct integral decomposition (see, e.g., Reed-Simon~\cite{RS4}) following Gel'fand~\cite{Gelfand}.
It describes the behavior of electrons or holes in a one-dimensional crystal (e.g., Kittel~\cite{Kittel}).
Each closed interval of the spectrum is called a {\it spectral band}.
In condensed matter physics, it is essential to determine the location of the bands, especially the band-to-band gap ({\it band gap}).
Estimating the band gap is applied to study the electrical properties of semiconductors.

Even if $v$ is not real, $H$ is still a bounded operator. 
Thus, while its spectrum $\sigma(H)$ is a compact set on the complex plane, its shape often becomes more complicated.
In the case of self-adjoint operators, a nice theory, including the spectral decomposition theorem, 
can be applied. 
On the other hand, for non-self-adjoint cases, no such nice general theory exists and 
must be analyzed on a problem-by-problem basis.
Because of this inconvenience, the spectral theory for complex-valued potential literature has been relatively few.

Non-self-adjoint Schr\"odinger-type operators (non-Hermitian Hamiltonian in physics) 
have emerged naturally in $\mathcal{PT}$-symmetric (parity-time symmetric) quantum theory (see, e.g., Bender~\cite{Bender}), 
providing a strong incentive for their study.
For differential operators with complex-valued periodic potentials, Valiev~\cite{OV} 
is an excellent guide for researchers in this field. 
In the discrete (Jacobi matrix) case, a theory similar to that of differential operators can be constructed. 
Moreover, the direct and inverse spectral theory for more general Jacobi matrices with complex periodic coefficients has also been obtained (Hochstadt~\cite{Hochstadt}, Papanicolaou~\cite{PapaJOP}).
An example of the significant difference between complex-valued periodic potentials and real-valued periodic potentials is given by Gasimov~\cite{Gasimov}(see also Guillemin and Uribe~\cite{GU}).
In continuous case, if $V(x)\in L^2_{\loc}(\mathbb{R})$ is real-valued, a famous theorem of Borg~\cite{Borg}, that is, $\sigma(-\frac{d^2}{dx^2}+V)=[0, \infty)$ if and only if $V(x)=0$ a.e.. 
In the case of complex-valued periodic potentials, 
the results are very different: Gasimov~\cite{Gasimov} showed that if 
$$
V(x) = \sum_{k=1}^{\infty}c_ke^{ikx},\quad \sum_{k=1}^{\infty}|c_k|<\infty,
$$
then $\sigma(-\frac{d^2}{dx^2}+V)=[0, \infty)$.
Papanicolaou~\cite{PapaJOP} showed a discrete version of this Gasimov's theorem. 

The discriminant (defined in Section \ref{sec:discriminant}) determines the spectrum of the one-dimensional Schr\"odinger operator with periodic potential.
The discriminant is a polynomial in $E$ with the period of the potential as its degree.
Papanicolaou~\cite{PapaJOP} proved there exists at most $L!$ different $L$-periodic potentials whose discriminants are the same. As a result, the spectra of these operators coincide.
Papanicolaou~\cite{PapaJOP} used the fact that the coefficients of the discriminant are represented by the elementary symmetric polynomials of $V(1), V(2), \ldots, V(L)$ to show this result.
This representation is suitable for inverse problems but not for the perturbation theory of discriminant because of the difficulty in obtaining information about the value of the discriminant.

The most fundamental problem in the study of operators is determining the operator's spectrum. 
As described in Section \ref{sec:discriminant}, the problem of determining the spectrum for a periodic potential with period $L$ is equivalent to the problem of finding the intersection of two algebraic curves of degree $L$. It becomes more difficult as $L$ increases; thus, a few examples have been studied in detail. The sparse potentials treated in this paper can be analyzed even for large $L$.
This paper presents a Chebyshev polynomial~(defined in Section \ref{sec:Chebyshev}) representation of the discriminant, which is then applied to the spectrum analysis.

This paper is organized as follows.  
Section 2 defines the discriminant, states that the discriminant describes the spectrum of $H$, and introduces the Floquet spectrum.
Section 3 presents the Chebyshev polynomial representation of the discriminant, which is the main theorem of this paper, after introducing Chebyshev polynomials and listing their properties. Furthermore, we apply this theorem to derive two properties concerning integrals in $[-2, 2]$ of the discriminant.
Section 4 derives the first-order Taylor polynomials of the discriminant using the main theorem.
Section 5 shows that exactly $L$ spectral bands appear for nonzero real $v$ and that the band outside $[-2, 2]$ converges to a point as $L$ is large, applying the result of Section 4.
Section 6 studies the Floquet spectrum by perturbation method for $v\in\mathbb{C}$ with small and large $|v|$.

\section{Discriminant and Spectrum of $H$}\label{sec:discriminant}
The spectrum of the one-dimensional discrete Schr\"odinger operator with 
periodic potential is represented by the {\it Hill discriminant}. 

For every $E\in\mathbb{C}$, the equation $Hu = Eu$ can be uniquely solved by giving initial values $u(0)$ and $u(1)$.
Precisely, the solution $u(n)$ can be represented as
$$
\left(\begin{array}{c}
u(n+1) \\
u(n)
\end{array}
\right) = \Phi_n(E)\left(\begin{array}{c}
u(1) \\
u(0)
\end{array}
\right), 
$$
where
$$
\Phi_n(E) = 
\left\{\begin{array}{ll}
A_n(E)\cdots A_1(E) & (n\geq 1) \\
I & (n=0) \\ 
A_{n+1}(E)^{-1}\cdots A_0(E)^{-1} & (n\leq -1)
\end{array}\right.
$$
and 

$$
A_n(E) = \left(\begin{array}{cc}
E - V(n) & - 1 \\
1 & 0
\end{array}\right).
$$
Note that $\det\Phi_n(E) = 1$ since $\det A_n(E) = 1$. 
The spectrum $\sigma(H)$ of $H$ can be characterized as
$$
\sigma(H) = \{E\in\mathbb{C}: \Delta_L(E) \in [-2, 2] \}, 
$$
where $\Delta_L(E) = \tr\Phi_L(E)$ is the Hill discriminant (``discriminant'' for short) of $H$. 
Note that $\Delta_L(E)$ is a monic polynomial. 
This result was proved by Rofe-Beketov~\cite{Rofe} for continuous Schr\"odinger operators with complex coefficients, but it can also be established for the discrete case.
We note that in the case $v=0$, i.e., free Laplacian $H_0$, $\Delta_L(E) = 2T_L(E/2)$ for the first kind Chebyshev polynomial $T_L$ of degree $L$~(described in Section \ref{sec:Chebyshev}). 
In this case, $\Delta_L(E) = 2\cos (L\cos^{-1}(E/2))$ by the definition of the Chebyshev polynomial, which shows that $\sigma(H_0) = [-2, 2]$.
We introduce the {\it Floquet spectrum} 
$$
\sigma_{\kappa}(H) = \{E\in\mathbb{C}: \Delta_L(E) = 2\cos\kappa \}
$$
for $\kappa\in [0, \pi]$, then $\sigma(H)$ can be written as the disjoint union
\begin{equation}\label{eq:union}
\sigma(H) = \bigcup_{\kappa\in [0, \pi]}\sigma_{\kappa}(H).
\end{equation}
$\Delta_L(E) = 2\cos\kappa$ is equivalent to $\Phi_L(E)$ having eigenvalues $e^{i\kappa}, e^{-i\kappa}$. 
Since the discriminant $\Delta_L(E)$ is a polynomial of degree $L$, 
$\Delta_L(E)-2\cos\kappa$ has $L$ roots $E_1(\kappa), \ldots, E_L(\kappa)$ with multiplicity(these are called $\kappa$-{\it Floquet eigenvalues}), 
each $E_j(\kappa)$ of which is continuous with respect to $\kappa$. 
Therefore, (\ref{eq:union}) leads that $\sigma(H)$ consists of $L$ closed bounded analytic arcs lying in the complex plane at most.

\begin{figure}[htbp]
 \centering
  \includegraphics[width=100mm]{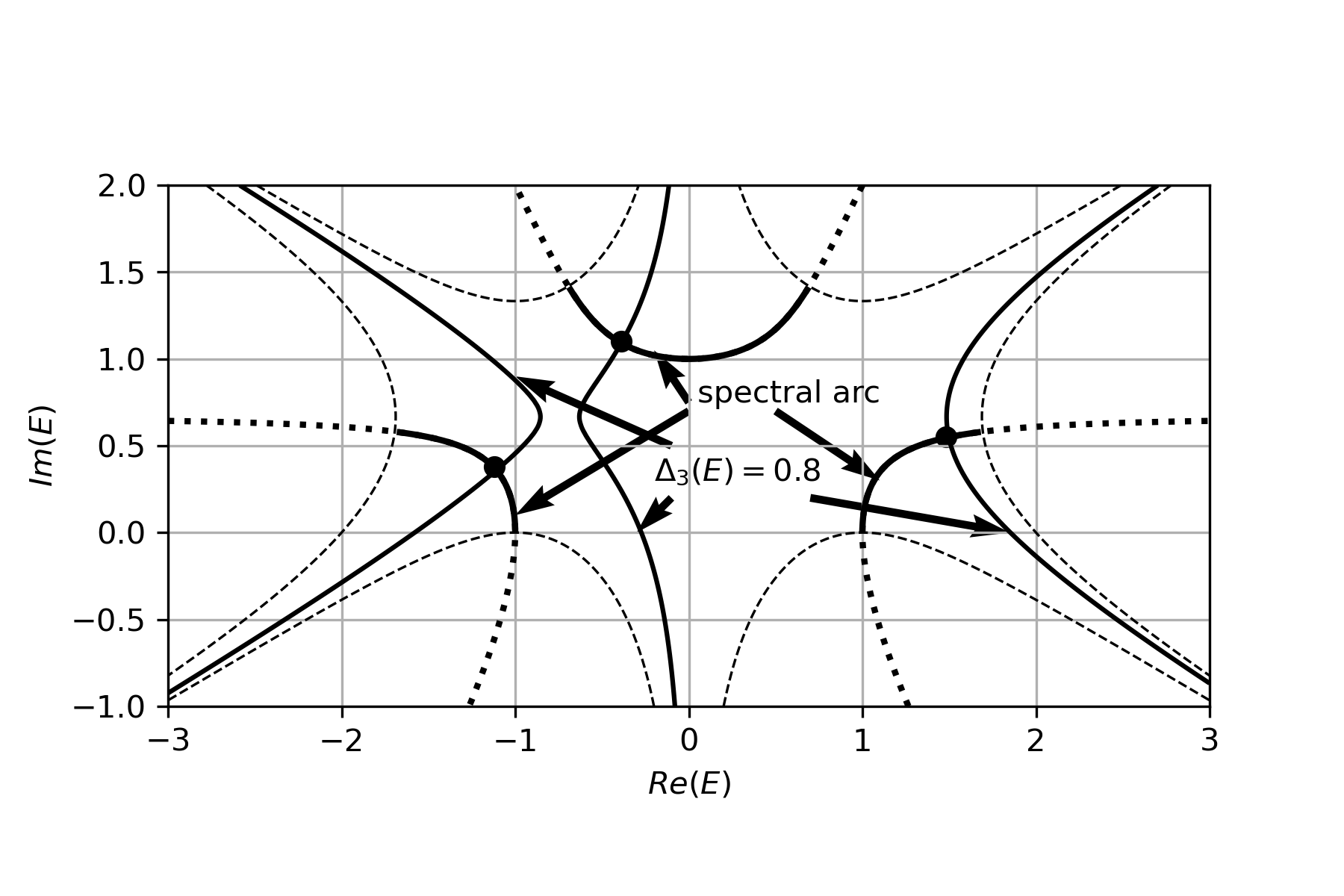}
 \caption{The spectrum for $L=3$ and $v=2i$}
 \label{fig:L3a2explain}
\end{figure}

\begin{example}\label{ex:L3}
By setting $E=x+iy~(x, y)\in\mathbb{R}^2$, 
the Floquet spectrum $\sigma_{\kappa}(H)$ can be represented as the intersection of the curves 
$\Re\Delta_L(E) = \Re\Delta_L(x+iy) = 2\cos\kappa$ and $\Im\Delta_L(E) = \Im\Delta_L(x+iy) = 0$.
These two curves are real zeros of polynomials of total degree $L$ in $x$ and $y$.
The two curves represented by the real and imaginary parts of a holomorphic function on the complex plane intersect perpendicularly. 
Therefore, $\Re\Delta_L(x+iy) = 2\cos\kappa$ and $\Im\Delta_L(x+iy)=0$ have perpendicular intersections.
For example, consider the case $L=3$ and $v=2i$. In this case,  
$\Re\Delta_3(x+iy) = x^3 - 3xy^2 - 3x + 4xy$ and 
$\Im\Delta_3(x+iy) = 3x^2y-y^3-3y - 2(x^2-y^2-1)$.
In Figure \ref{fig:L3a2explain}, the dashed line represents $\Re\Delta_3(x+iy) = \pm 2$ 
and the dotted line represents $\Im\Delta_3(x+iy) = 0$. 
The part of the dotted line that is solid is the arc of the spectrum, 
and the three curves that intersect these curves are $\Re\Delta_3(x+iy) = 0.8~(\kappa=\cos^{-1}0.4)$. 
These intersections (three filled circles in Figure \ref{fig:L3a2explain}) correspond to the Floquet spectrum $\sigma_{\cos^{-1}0.4}(H)$.
\end{example}

\begin{rem}
It is easy to verify that the spectrum is symmetric about the imaginary axis for pure imaginary $v$.
\end{rem}

\section{Representation of the discriminant using Chebyshev polynomials}\label{sec:Chebyshev}
In this section, we list several formulas for the Chebyshev polynomials (see, e.g., Mason-Handscomb~\cite{Chabyshev}) to be used later, 
and represent the discriminant  $\Delta_L(E)$ using Chebyshev polynomials.
\begin{defn}\label{defn:Chebyshev}
The Chebyshev polynomials of the first kind $T_n(x)$ are defined by
$T_{n}(\cos\theta)=\cos(n\theta)~(n\geq 0)$.
Similarly, define the Chebyshev polynomials of the second kind $U_{n}(x)$ are defined by
$U_{n}(\cos\theta)\sin\theta=\sin(n+1)\theta~(n\geq 0)$.
\end{defn}
It is easy to verify that $T_n(-x) = (-1)^nT_n(x)$ and $U_n(-x) = (-1)^nU_n(x)$.
The first few Chebyshev polynomials of the first and second kind are $T_0(x)=1$, $T_1(x)=x$, $T_2(x)=2x^2-1$, $T_3(x)=4x^3-3x$, $U_0(x)=1$, $U_1(x)=2x$, $U_2(x)=4x^2-1$, $U_3(x)=8x^3-4x$. 
It is also easily seen that 
\begin{eqnarray*}
T_n(\cosh\xi) &=& \cosh n\xi, \\ 
U_n(\cosh\xi) &=& \frac{\sinh(n+1)\xi}{\sinh\xi}.
\end{eqnarray*}
\begin{prop}\label{prop:identity}
The Chebyshev polynomials satisfy the following recursive relations:
\begin{eqnarray*}
T_{n+1}(x) &=& 2xT_n(x) - T_{n-1}(x) ~(n\geq 1), \\
U_{n+1}(x) &=& 2xU_n(x) - U_{n-1}(x) ~(n\geq 1),  \\
T_n(x) &=& \frac{1}{2}(U_n(x)-U_{n-2}(x)) ~(n\geq 2).
\end{eqnarray*}
\end{prop}

\begin{prop}\label{prop:identity-diff}
The following formulas hold for the derivative of Chebyshev polynomials.
\begin{eqnarray*}
T'_n(x) &=& nU_{n-1}(x) ~(n\geq 1), \\
U'_n(x) &=& \frac{(n+1)T_{n+1}(x) - xU_n(x)}{x^2-1} ~(n\geq 0)
\end{eqnarray*}
\end{prop}

\begin{prop}\label{prop:orthogonality}
Both $\{T_n(x)\}_{n=0}^{\infty}$ and $\{U_n(x)\}_{n=0}^{\infty}$ form a sequence of orthogonal polynomials in the following sense:
\begin{eqnarray*}
\frac{2}{\pi}\int_{-1}^1T_n(x)T_m(x)\frac{1}{\sqrt{1-x^2}}dx &=& \delta_{nm} \\
\frac{2}{\pi}\int_{-1}^1U_n(x)U_m(x)\sqrt{1-x^2}dx &=& \delta_{nm}, 
\end{eqnarray*}
where $\delta_{nm}$ is the Kronecker Delta. 
\end{prop}

\begin{prop}\label{prop:Chebyshevzeros}
\begin{eqnarray*}
T_n(x) &=& 2^{n-1}\prod_{k=1}^n\left(x-\cos\frac{(2k-1)\pi}{2n}\right), \\
U_n(x) &=& 2^{n}\prod_{k=1}^n\left(x-\cos\frac{k\pi}{n+1}\right)
\end{eqnarray*}
\end{prop}

\begin{lem}\label{lem:poermatrix}
Let $\Phi_0(E)$ be the matrix  
$$
\Phi_0(E) =\left(\begin{array}{cc}
E & -1 \\
1 & 0
\end{array}
\right),
$$
then, $\tr(\Phi_0(E)^n) = 2T_n(E/2)$. 
\end{lem}

\begin{proof}
We first consider the case of $\Phi_0(E)$ has eigenvalues $\lambda, \lambda^{-1}$ satisfying $\lambda^2\ne 1$. 
By direct computation, we learn
\begin{equation}\label{eq:power}
\Phi_0(E)^n = \frac{\lambda^n-\lambda^{-n}}{\lambda-\lambda^{-1}}\Phi_0(E) - \frac{\lambda^{n-1}-\lambda^{-(n-1)}}{\lambda-\lambda^{-1}}I.
\end{equation}
Let $\lambda=e^{i\theta}$, i.e., $E=2\cos\theta$~($\theta$ need not be a real number). 
Taking trace of both sides of (\ref{eq:power}), combining the second and the third identities in Proposition \ref{prop:identity} yields
\begin{eqnarray*}
\tr(\Phi_0(E)^n) &=& \frac{e^{in\theta}-e^{-in\theta}}{e^{i\theta}-e^{-i\theta}}E - 2\frac{e^{i(n-1)\theta}-e^{-i(n-1)\theta}}{e^{i\theta}-e^{-i\theta}} \\
&=& EU_{n-1}(E/2) - 2U_{n-2}(E/2) \\
&=& U_n(E/2) - U_{n-2}(E/2) = 2T_n(E/2).
\end{eqnarray*}
We consider the case $\lambda^2=1$, i.e., $E=\pm 2$.  
In this case, $\lambda$ is a double root of the eigenpolynomial of $\Phi_0(E)$, then we have
\begin{equation}\label{eq:power2}
\Phi_0(E)^n = n\lambda^{n-1}\Phi_0(E) - (n-1)\lambda^nI.
\end{equation}
Taking trace of both sides of (\ref{eq:power2}), we have
$$
\tr(\Phi_0(E)^n) = n(\pm 1)^{n-1}(\pm 2) - 2(n-1)(\pm 1)^n.
$$
Here is the double sign in the same order.
Then, $\tr(\Phi_0(E)^n)=2$ for $E=2$, and $=2(-1)^n$ for $E=-2$, these are consistent with $2T_n(E/2)$.
\end{proof}

\noindent
Then, our main theorem is stated as follows:
\begin{thm}\label{thm:represent}
The discriminant $\Delta_L(E)$ can be represented as
\begin{equation}\label{eq:mainrepresentation}
\Delta_L(E) = 2T_L(E/2) - vU_{L-1}(E/2)
\end{equation}
for $L\geq 1$.
\end{thm}
\begin{proof}
From the cyclic invariance of the matrix trace, the discriminant $\Delta_L(E)$ 
equals the trace of the matrix given by
\begin{equation}\label{eq:basic}
\Phi_0(E)^{L-1}
\left(\begin{array}{cc}
E - v& -1 \\
1 & 0
\end{array}
\right)=\Phi_0(E)^L
-v\Phi_0(E)^{L-1}
\left(\begin{array}{cc}
1& 0 \\
0 & 0
\end{array}
\right).
\end{equation}

\noindent
For $L=1, 2$, by direct computation, we have
\begin{eqnarray*}
\Delta_1(E) &=& E - v = 2T_1(E/2) - vU_0(E/2), \\
\Delta_2(E) &=& E^2 - 2 - vE = 2T_2(E/2) - vU_1(E/2).
\end{eqnarray*}
Next, assume $L\geq 3$, 
using (\ref{eq:power}), we have
$$
\Phi_0(E)^{L-1}
\left(\begin{array}{cc}
1& 0 \\
0 & 0
\end{array}
\right) = 
\frac{\lambda^{L-1}-\lambda^{-(L-1)}}{\lambda-\lambda^{-1}}\left(\begin{array}{cc}
E & 0 \\
0 & 0
\end{array}
\right) - \frac{\lambda^{L-2}-\lambda^{-(L-2)}}{\lambda-\lambda^{-1}}\left(\begin{array}{cc}
1& 0 \\
0 & 0
\end{array}
\right).
$$
Taking trace in both sides shows that 
\begin{eqnarray}
&& \tr\left\{\Phi_0(E)^{L-1}
\left(\begin{array}{cc}
1& 0 \\
0 & 0
\end{array}
\right)\right\} \nonumber \\
&=& \frac{\lambda^{L-1}-\lambda^{-(L-1)}}{\lambda-\lambda^{-1}}E-\frac{\lambda^{L-2}-\lambda^{-(L-2)}}{\lambda-\lambda^{-1}} \nonumber \\
&=& EU_{L-2}(E/2) - U_{L-3}(E/2) = U_{L-1}(E/2). \label{eq:U}
\end{eqnarray}
From Lemma \ref{lem:poermatrix}, we learn that $\tr(\Phi_0(E)^L) = 2T_L(E/2)$. 
Substituting $\tr(\Phi_0(E)^L) = 2T_L(E/2)$ and (\ref{eq:U}) into (\ref{eq:basic}),  
we have
\begin{eqnarray*}
\Delta_L(E) &=& \tr(\Phi_L(E))-v\tr\left\{\Phi_0(E)^{L-1}
\left(\begin{array}{cc}
1& 0 \\
0 & 0
\end{array}
\right)\right\} \\
&=& 2T_L(E/2) - vU_{L-1}(E/2).
\end{eqnarray*}
\end{proof}
\begin{cor}\label{cor:zeros}
Let  $\beta_{k}$ be the $k$-th root of $U_{L-1}(E/2)$ is given by 
$$
\beta_{k} = 2\cos\frac{k\pi}{L}\quad (k=1, 2, \ldots, L-1).
$$
Then, $\{\beta_1, \beta_2, \ldots, \beta_{L-1}\}\subset\sigma(H)$. Moreover, 
$2\in\sigma(H)$ if and only if $v\in [0, 4/L]$, also, $-2\in\sigma(H)$ if and only if $v\in[-4/L, 0]$. 
In particular, $\pm 2\not\in\sigma(H)$ if $v\not\in\mathbb{R}$. 
\end{cor}
\begin{proof}
Substituting $\beta_{k}$ into (\ref{eq:mainrepresentation}) yields 
$\Delta_L(\beta_{k}) = 2\cdot (-1)^{k}\in [-2, 2]$. 
Furthermore, substituting $E = \pm 2$ into (\ref{eq:mainrepresentation}) yields
$\Delta_L(2) = 2 - Lv$ and $\Delta_L(-2) = (-1)^L(2 + Lv)$.
$2\in\sigma(H)$ if and only if $\Delta_L(2) = 2 - Lv\in[-2, 2]$, which means that $v\in [0, 4/L]$. 
Similarly, it can be shown that $-2\in\sigma(H)$ and that $v\in[-4/L, 0]$ is equivalent.
\end{proof}

Corollary \ref{cor:zeros} shows that the roots $\beta_1, \beta_2, \ldots, \beta_{L-1}$ of $\Delta_L(E)$ 
that corresponding to one of the endpoints of each spectral arc is concentrated at endpoints $E = \pm 2$ 
of the spectrum of the free Laplacian as $L$ is large, 
and their density is approximately $\displaystyle\frac{1}{\pi\sqrt{4-E^2}}$ on the real axis.

\begin{cor}\label{rem:osc}
\begin{equation}\label{eq: integral}
\int_{-2}^2\Delta_L(E)dE = \left\{\begin{array}{cl}
-\frac{4}{L^2-1} & \mbox{$L$: even} \\
-\frac{2v}{L} & \mbox{$L$: odd}.
\end{array}
\right.
\end{equation}
\end{cor}
\begin{proof}
Using (\ref{eq:mainrepresentation}) and the definition of the Chebyshev polynomial, and 
substituting $E=2\cos\theta$, the integral then becomes 
\begin{eqnarray*}
\int_{-2}^2\Delta_L(E)dE &=& \int_{-2}^2(2T_L(E/2) - vU_{L-1}(E/2))dE \\
&=& \int_0^{\pi}(2\cos L\theta\sin\theta - v\sin L\theta)d\theta.
\end{eqnarray*}
This integral can be easily computed to obtain the desired result.
\end{proof}

This integral (\ref{eq: integral}) approaches $0$ as $L$ is large, which means that $\Delta_L(E)$ oscillates intensively within 
$[-2, 2]$ as large $L$.
\begin{thm}\label{thm:Parseval}
\begin{eqnarray*}
\frac{1}{2\pi}\int_{-2}^2|\Delta_L(E)|^2\sqrt{4-E^2}dE &=& 2 + |v|^2 \quad(L\geq 2) \\
\frac{1}{2\pi}\int_{-2}^2|\Delta_1(E)|^2\sqrt{4-E^2}dE &=& 1 + |v|^2 
\end{eqnarray*}
\end{thm}
\begin{proof}
For $L=1$, direct computation of the integral yields
\begin{eqnarray*}
&& \frac{1}{2\pi}\int_{-2}^2|\Delta_1(E)|^2\sqrt{4-E^2}dE \\
&=& \frac{1}{2\pi}\int_{-2}^2|E-v|^2\sqrt{4-E^2}dE \\
&=& \frac{1}{2\pi}\int_{-2}^2(E^2-(v+\overline{v})E + |v|^2)\sqrt{4-E^2}dE \\
&=& \frac{1}{2\pi}\int_{-2}^2(E^2 + |v|^2)\sqrt{4-E^2}dE = 1 + |v|^2.
\end{eqnarray*}
We consider the case $L\geq 2$. 
Theorem \ref{thm:represent} and the third identity in Proposition \ref{prop:identity} together shows that 
\begin{eqnarray*}
\Delta_L(E) &=& 2T_L(E/2) - vU_{L-1}(E/2) \\
&=& U_L(E/2)-U_{L-2}(E/2) - vU_{L-1}(E/2). 
\end{eqnarray*}
Using Proposition \ref{prop:orthogonality} and Parseval's identity, one can show 
\begin{eqnarray*}
&& \frac{1}{2\pi}\int_{-2}^2|\Delta_L(E)|^2\sqrt{4-E^2}dE \\
&=& \frac{1}{2\pi}\int_{-2}^2|U_L(E/2)|^2\sqrt{4-E^2}dE + \frac{1}{2\pi}\int_{-2}^2|U_{L-2}(E/2)|^2\sqrt{4-E^2}dE \\ 
&& + |v|^2\frac{1}{2\pi}\int_{-2}^2|U_{L-1}(E/2)|^2\sqrt{4-E^2}dE= 2 + |v|^2.
\end{eqnarray*}
\end{proof}

The integral in Theorem \ref{thm:Parseval} does not depend on $L(\geq 2)$, but only on $v$. 
Theorem \ref{thm:Parseval} that the value of the weight $\sqrt{4-E^2}$ 
becomes smaller near $\pm 2$ indicates that as $L$ becomes larger, 
$|\Delta_L(E)|$ becomes larger near $\pm 2$. 
From this fact and Corollary \ref{rem:osc}, we learn that $\Delta_L(E)$ oscillates intensively near $\pm 2$ for large $L$.
\section{The first-order Taylor approximation of the discriminant}
In this section, we give the first-order Taylor polynomials at $\alpha_j~(j=1, 2, \ldots, L)$, 
$\beta_{k}~(k=1, 2, \ldots, L-1)$, and $\pm 2$ on the real axis of the discriminant for application to the analysis of the spectrum of $H$ near $[-2, 2]$.

\begin{lem}\label{lem:Taylor}
The first-order Taylor polynomials at $\alpha_j~(j=1, 2, \ldots, L)$, 
$\beta_{k}~(k=1, 2, \ldots, L-1)$, and $\pm 2$ are given by:
{\small
\begin{eqnarray}
\Delta_L(E) &=& \frac{v(-1)^j}{\sin\theta_j} + \frac{2L\sin^2\theta_j-v\cos\theta_j}{2\sin^3\theta_j}(E-\alpha_j) + O((E-\alpha_j)^2) \label{eq:alpha}\\
\Delta_L(E) &=& 2(-1)^{k} + \frac{v(-1)^{k}}{2\sin^3\phi_{k}}(E-\beta_{k}) + O((E-\beta_{k})^2), \label{eq:beta} \\
\Delta_L(E) &=& 2-vL + \left\{L^2-\frac{v}{4}L(L^2-1)\right\}(E-2) + O((E-2)^2) \label{eq:2}\\
\Delta_L(E) &=& (-1)^L(2+vL) + (-1)^{L}\left\{L^2+\frac{v}{4}L(L^2-1)\right\}(E+2) \nonumber \\ 
&& + O((E+2)^2),\label{eq:-2}
\end{eqnarray}
}
\noindent
where $\alpha_j = 2\cos\theta_j = 2\cos\frac{(2j-1)\pi}{2L}$ and $\beta_{k} = 2\cos\phi_{k}=2\cos\frac{k\pi}{L}$. 
\end{lem}
\begin{proof} Taking derivative of (\ref{eq:mainrepresentation}) and Proposition \ref{prop:identity-diff} together shows that
\begin{eqnarray}
\Delta_L'(E) &=& T_L'(E/2) - \frac{v}{2}U_{L-1}'(E/2) \nonumber \\
&=& LU_{L-1}(E/2) - \frac{v}{2}\cdot\frac{LT_L(E/2) - \frac{E}{2}U_{L-1}(E/2)}{(E/2)^2-1} \nonumber \\
&=& LU_{L-1}(E/2) - v\frac{2LT_L(E/2)-EU_{L-1}(E/2)}{E^2-4}.\label{eq: ChevDeri2} 
\end{eqnarray}
Since $\alpha_j$ satisfies $T_L(\alpha_j/2)=0$ and $\beta_{k}$ satisfies $U_{L-1}(\beta_{k}/2)=0$, 
substituting $\alpha_j$ and $\beta_{k}$ into (\ref{eq: ChevDeri2}) yields (\ref{eq:alpha}) and (\ref{eq:beta}).

Substituting $E = 2\cos\theta~(0<\theta<\pi)$ into (\ref{eq: ChevDeri2}), we have
\begin{eqnarray}
\Delta_L'(2\cos\theta) &=& LU_{L-1}(\cos\theta) - v\frac{2LT_L(\cos\theta)-2\cos\theta U_{L-1}(\cos\theta)}{4\cos^2\theta-4} \nonumber \\
&=& L\frac{\sin L\theta}{\sin\theta} + v\frac{2L\cos L\theta - 2\cos\theta\frac{\sin L\theta}{\sin\theta}}{4\sin^2\theta} \nonumber \\
&=& L\frac{\sin L\theta}{\sin\theta} + v\frac{L\sin\theta\cos L\theta - \cos\theta\sin L\theta}{2\sin^3\theta}. \label{eq:limit2}
\end{eqnarray}
Taking limit $\theta\to 0$ of both sides of (\ref{eq:limit2}), we have
\begin{equation}\label{eq:der2}
\Delta_L'(2) = L^2 -\frac{v}{4}L(L^2-1).
\end{equation}
Combining (\ref{eq:der2}) and $\Delta_L(2)=2-vL$ yields (\ref{eq:2}). 
Using 
\begin{eqnarray*}
\Delta_L(-E) &=& 2T_L(-E/2)-vU_{L-1}(-E/2) \\
&=& 2(-1)^LT_L(E/2) - v(-1)^{L-1}U_{L-1}(E/2),
\end{eqnarray*}
we have 
$$
\Delta_L'(-2\cos\theta) = L(-1)^L\frac{\sin L\theta}{\sin\theta} + v(-1)^{L-1}\frac{L\sin\theta\cos L\theta - \cos\theta\sin L\theta}{2\sin^3\theta}
$$
and then, taking limit $\theta\to 0$, we have
\begin{equation}\label{eq:der-2}
\Delta_L'(-2) = (-1)^LL^2 - (-1)^{L-1}\frac{v}{4}L(L^2-1).
\end{equation}
Combining (\ref{eq:der-2}) and $\Delta_L(-2)=(-1)^L(2+vL)$ yields (\ref{eq:-2}). 
\end{proof}
\begin{rem}
Since (\ref{eq: ChevDeri2}) is the derivative of polynomial $\Delta_L(E)$, (\ref{eq: ChevDeri2}) is also a polynomial, therefore, $2LT_L(E/2)-EU_{L-1}(E/2)$ becomes a divisor of $E^2-4$.
For example, in the case of $L=3$, it is equals to $2\cdot 3T_3(E/2)-EU_{2}(E/2) = 2E^3-8E=2E(E^2-4)$. 
\end{rem}

\begin{rem}
From (\ref{eq:alpha}) in Lemma \ref{lem:Taylor}, we have
$$
\Delta_L(\alpha_j) = \frac{v(-1)^j}{\sin\frac{(2j-1)\pi}{2L}}.
$$
This indicates that the discriminant's amplitude increases almost in proportion to $Lv$ near both ends of $[-2, 2]$. It is consistent with Corollary \ref{rem:osc} and Theorem \ref{thm:Parseval}.
\end{rem}
\section{Spectrum for real $v$}
In this section, we consider the case of real $v$ compared to where $v$ is not real.
$H$ is a self-adjoint operator if $v$ is real, therefore, the spectrum $\sigma(H)\subset\mathbb{R}$. 
Throughout this section, interval means an interval on the real axis.
\begin{thm}(A part of Theorem 4.10 in Kato~\cite{Kato}) \label{thm:Kato}
Let $T$ be self-adjoint, and $A$ bounded symmetric. Then $S=T+A$ is self-adjoint and 
$$
\sup_{\zeta\in\sigma(S)}\dist(\zeta, \sigma(T))\leq\|A\|,
$$
where $\dist(\zeta, D) = \inf\{|\zeta-x| : x\in D\}$.  
\end{thm}
Theorem \ref{thm:Kato} implies the roughest estimate of the location of the spectrum as follows.
\begin{prop}
$$
\sup_{\zeta\in\sigma(H)}\dist(\zeta, [-2, 2])\leq |v|,
$$
\noindent
that is, $\sigma(H)\subset [-2-|v|, 2 + |v|]$. 
\end{prop}

Since $\Delta_L(E)$ is a polynomial of degree $L$, the number of extreme values of the discriminant is at most $L-1$; therefore, $\sigma(H)=\{E\in\mathbb{R}:\Delta_L(E)\in[-2, 2]\}$ consists of at most $L$ closed intervals. 
The structure of the spectrum can be further detailed, as in the following Theorem \ref{thm:realstructure}.

\begin{thm} \label{thm:realstructure}
If $v\in\mathbb{R}$ is nonzero, $\sigma(H)$ consists of exactly $L$ closed intervals. 
\end{thm}

\begin{proof}
Let $\beta_{k}~(k=1, 2, \ldots, L-1)$ be defined in Lemma \ref{lem:Taylor}, then 
$\Delta_L(\beta_{k}) = 2(-1)^{k}$.
Since the first-order coefficient of (\ref{eq:beta}), i.e.,
$$
\eta_{k} = \frac{v(-1)^{k}}{2\sin^3\frac{k\pi}{L}}
$$
is never zero, therefore, 
$\Delta_L(E)$ intersects $2(-1)^{k}$ transversely at $E=\beta_{k}~(k=1, 2, \ldots, L-1)$.
This result means that $\beta_{k}$ is an endpoint of a spectral band.
Since it can be shown the same way for $v<-2$ or $L$ is even, 
we will only prove the $v>2$ and odd $L$ case here.

Since $\Delta_L(-2) = -(2+Lv)<-2$ and $\Delta_L(\beta_{L-1})=2$, there exists $\gamma_{L-1}\in (-2, \beta_{L-1})$ such that $\Delta_L(\gamma_{L-1})=-2$. 
Therefore, $[\gamma_{L-1}, \beta_{L-1}]$ contains at least one spectral band $I_{L-1}$ of $\sigma(H)$. 
Next, since $\Delta_L'(\beta_{L-1})=\eta_{L-1}>0$ and $\Delta_L(\beta_{L-2}) = -2$, 
there exists $\gamma_{L-2}\in (\beta_{L-1}, \beta_{L-2})$ and $[\gamma_{L-2}, \beta_{L-2}]$ 
contains at least one spectral band $I_{L-2}$.
In the same way, we show that there exists $\gamma_{L-3}\in (\beta_{L-2}, \beta_{L-3}), \cdots,  \gamma_{1}\in (\beta_{2}, \beta_1)$ such that all $[\gamma_{k}, \beta_{k}]$ each contain a spectral band $I_{k}$ for every $k = 1, 2, \ldots, L-1$.
Since $\Delta_L(E)$ is a monic polynomial, there exists $\beta_0>\beta_1$ such that $\Delta_L(\beta_0) = 2$. 
Therefore, since $\Delta_L(\beta_1) = -2$, $\Delta_L'(\beta_1) = \eta_1<0$, there exists $\gamma_0\in (\beta_1, \beta_0)$ such that $\Delta_L(\gamma_0) = -2$, that is, $[\gamma_0, \beta_0]$ cotains at least a spectral band $I_0$. 
From the above, there exist at least $L$ spectral bands. 
Since there exist at most $L$ spectral bands, this indicates that $I_{k}=[\gamma_{k}, \beta_{k}]~(k = 0, 1, \ldots, L-1)$.
Therefore, we conclude 
$$
\sigma(H) = \bigcup_{k=0}^{L-1}[\gamma_{k}, \beta_{k}].
$$ 
Thus, we complete the proof.
\end{proof}

By carefully observing the proof of Theorem \ref{thm:realstructure} 
and the value at $E = \pm 2$ of the discriminant $\Delta_L(E)$, 
the following result is easily derived.
\begin{cor}\label{cor:bandlocation} 
Suppose that $\sigma(H) = \bigcup_{k = 0}^{L-1}I_{k}$, where $I_{k}=[\gamma_{k}, \beta_{k}]$ with $\gamma_{L-1}<\beta_{L-1}<\cdots<\gamma_0<\beta_0$. Then, $[-2, 2]\cap[\gamma_0, \beta_0]\ne\emptyset$ if and only if $0\leq v\leq 4/L$, and $[-2, 2]\cap[\gamma_{L-1}, \beta_{L-1}]\ne\emptyset$ if and only if $-4/L\leq v\leq 0$.
\end{cor}
We remark that Corollary \ref{cor:bandlocation} implies that for $|v|>4/L$, $[-2, 2]$ contains exactly $L-1$ spectral bands, and there exists only one band outside $[-2, 2]$.
If $L$ is large enough, $|v|>4/L$ is satisfied; moreover, we have a good approximation of the location of the spectral band outside $[-2, 2]$ as follows.
\begin{cor}
The spectral band $I_{*}$ outside $[-2, 2]$ closes a point 
either $\sqrt{4+v^2}$ for $v>0$, or $-\sqrt{4+v^2}$ as $L\to\infty$.
That is, 
$$
\sup_{E\in I_{*}}\dist(E, \{\sgn(v)\sqrt{4+v^2}\})\to 0\quad (L\to\infty),
$$ 
where $\sgn$ is the signum function.
\end{cor}
\begin{proof}
It is sufficient to show that all Floquet eigenvalues in the spectral band outside $[-2, 2]$ 
converge to either $\sqrt{4+v^2}$ for $v>0$, or $-\sqrt{4+v^2}$ for $v<0$ as $L\to\infty$.
By letting $E(\kappa)=2\cosh\xi(\kappa)~(\xi(\kappa)>0)$ be a Floquet eigenvalue larger than $2$, we obtain
\begin{eqnarray*}
\Delta_L(2\cosh\xi(\kappa)) &=& 2T_L(\cosh\xi(\kappa)) - vU_{L-1}(\cosh\xi(\kappa)) \\
&=& 2\cosh L\xi(\kappa) - v\frac{\sinh L\xi(\kappa)}{\sinh\xi(\kappa)} = 2\cos\kappa. 
\end{eqnarray*}
Thus, we have 
\begin{equation}\label{eq:limit}
2\coth L\xi(\kappa) - \frac{v}{\sinh\xi(\kappa)} = \frac{2\cos\kappa}{\sinh L\xi(\kappa)}.
\end{equation}
Taking limit $L\to\infty$ of both sides of (\ref{eq:limit}), we learn that $\sinh\xi(\kappa)\to v/2$ as $L\to\infty$, therefore, 
$E=2\cos\xi(\kappa)\to\sqrt{4+v^2}$ as $L\to\infty$. 
In the case $v<0$, $E\to -\sqrt{4+v^2}$ can be found similarly.
\end{proof}
\section{Floquet spectrum for small and large $v$}
In this section, we study perturbations of the Floquet spectrum for small and large $v$.
We first note that from Proposition \ref{prop:Chebyshevzeros}, all roots of $T_n(x)$ and $U_n(x)$ are simple.
\begin{thm}(Lagrange inversion formula, 3.6.6. in Abramowitz and Stegun\cite{AS})\label{thm:Lagrange}
Suppose $w=f(z)$ is analytic at a point $z_0$ and $f'(z_0)\ne 0$. 
Then $z=g(w)$ given by a power series has a nonzero radius of convergence:
$$
g(w) = z_0 + \sum_{n=1}^{\infty}\frac{g_n}{n!}(w-f(z_0))^n,
$$
where 
$$
g_n = \lim_{z\to a}\frac{d^{n-1}}{dz^{n-1}}\left\{\left(\frac{z-z_0}{f(z)-f(z_0)}\right)^n\right\}.
$$
\end{thm}
\begin{lem}\label{lem:Taylorapprox}
Let $f(z)$ be a polynomial with $\alpha$ as a simple root, and let $g(z)$ be a polynomial satisfying $g(\alpha)\ne 0$.    
Then, for $t \in\mathbb{C}$ with sufficiently small $|t|$, $f(z) - sg(z) - t$ has a simple root 
$\alpha(v, t)$ near $\alpha$ such that
$$
\alpha(v, t) = \alpha + \frac{g(\alpha)}{f'(\alpha)}s + \frac{t}{f'(\alpha)} + O((|s|+|t|)^2)
$$
\end{lem}
\noindent
{\it Sketch of the proof of Lemma \ref{lem:Taylorapprox} }

\noindent 
First, we note that $f'(\alpha)\ne 0$ since $\alpha$ is a simple root of $f(z)$. 
We consider the polynomial $h(z) = f(z)-sg(z)-t$. 
If $|s|$ and $|t|$ are small enough, $h(z)$ has a simple root $\alpha(s, t)$ near $\alpha$. 
Theorem \ref{thm:Lagrange} shows
\begin{equation}\label{eq:alphavt}
\alpha(s, t) = \alpha(s) + \frac{t}{f'(\alpha(s))-sg'(\alpha(s))} + O(|t|^2),
\end{equation}
where $\alpha(s)$ is a root of $f(z)-sg(z)=0$. 
$\alpha(s)$ is close enough to $\alpha$ for every sufficiently small $|s|$. 
Since $f(z)/g(z)$ is analytic at $\alpha$  the assumption, we can use Theorem \ref{thm:Lagrange} again, it shows
\begin{eqnarray}
\alpha(s) &=& \alpha + \frac{s}{\left.\left(\frac{f(z)}{g(z)}\right)'\right|_{z=\alpha}} + O(|s|^2) \nonumber \\
&=& \alpha + \frac{g(\alpha)}{f'(\alpha)}s + O(|s|^2) \label{eq:alphav}
\end{eqnarray}

Substituting (\ref{eq:alphav}) into (\ref{eq:alphavt}) and expanding the first-order term for $t$ using the geometric series for small $|s|$, we have
\begin{eqnarray*}
\alpha(s, t) &=& \alpha + \frac{g(\alpha)}{f'(\alpha)}s + \frac{t}{f'(\alpha(s))-sg'(\alpha(s))} + O(|s|^2) + O(|t|^2) \\
&=& \alpha + \frac{g(\alpha)}{f'(\alpha)}s + \frac{t}{f'(\alpha(s))}\sum_{n=0}^{\infty}\left(\frac{g'(\alpha(s))}{f'(\alpha(s))}\right)^ns^n + O(|s|^2) + O(|t|^2) \\
&=& \alpha + \frac{g(\alpha)}{f'(\alpha)}s + \frac{t}{f'(\alpha(s))} + O(|st|) + O(|s|^2) + O(|t|^2)
\end{eqnarray*}
\noindent
Since it follows from (\ref{eq:alphav}) that $f'(\alpha(s)) =f'(\alpha) + O(|s|)$, we have
$$
\frac{1}{f'(\alpha(s))} = \frac{1}{f'(\alpha)} + O(|s|).
$$
Summarize the above to reach the desired result. 
\qed

\begin{thm}\label{thm:localstructure}
Let $\alpha_j$ be defined in Lemma \ref{lem:Taylor}, and $\alpha_j(v, \kappa)$ be the element nearest $\alpha_j$ of the Floquet spectrum $\sigma_{\kappa}(H)$. For $v\in\mathbb{C}$ with sufficiently small $|v|$ and $\kappa\in [0, \pi]$ with sufficiently small $|\cos\kappa|$, 
$$
\alpha_j(v, \kappa) = \alpha_j + \frac{v}{L} + 2(-1)^{j-1}\frac{\sin\frac{(2j-1)\pi}{2L}}{L}\cos\kappa 
+ O((|v|+|\cos\kappa|)^2).
$$
\end{thm}
\begin{proof}
Consider the equation $\Delta_L(E) - 2\cos\kappa = 2T_L(E/2) - vU_{L-1}(E/2) - 2\cos\kappa = 0$.
By applying Lemma \ref{lem:Taylorapprox} to this equation, we have
\begin{equation}\label{eq:alphaperturb}
\alpha_j(v, \kappa) = \alpha_j + \frac{U_{L-1}(\alpha_j/2)}{T_L'(\alpha_j/2)}v + \frac{2\cos\kappa}{T_L'(\alpha_j/2)} + O(|v|^2+|\cos\kappa|^2).
\end{equation}
Proposition \ref{prop:identity-diff} leads 
$$
T_L'(\alpha_j/2)=LU_{L-1}(\alpha_j/2)=\frac{(-1)^{j-1}L}{\sin\frac{(2j-1)\pi}{2L}}.
$$ 
Substituting the above into (\ref{eq:alphaperturb}) completes the proof.
\end{proof} %
Theorem \ref{thm:localstructure} implies that $\sigma(H)$ has $L$ arcs approximately parallel to the real axis near $\alpha_j+v/L$ for $j = 1, 2, \ldots, L$.
Moreover, as $L$ increases, the potential becomes sparse, therefore, 
the spectrum close to the set $[-2, 2]$ on the real line.

\begin{figure}[htbp]
 \centering
  \includegraphics[width=100mm]{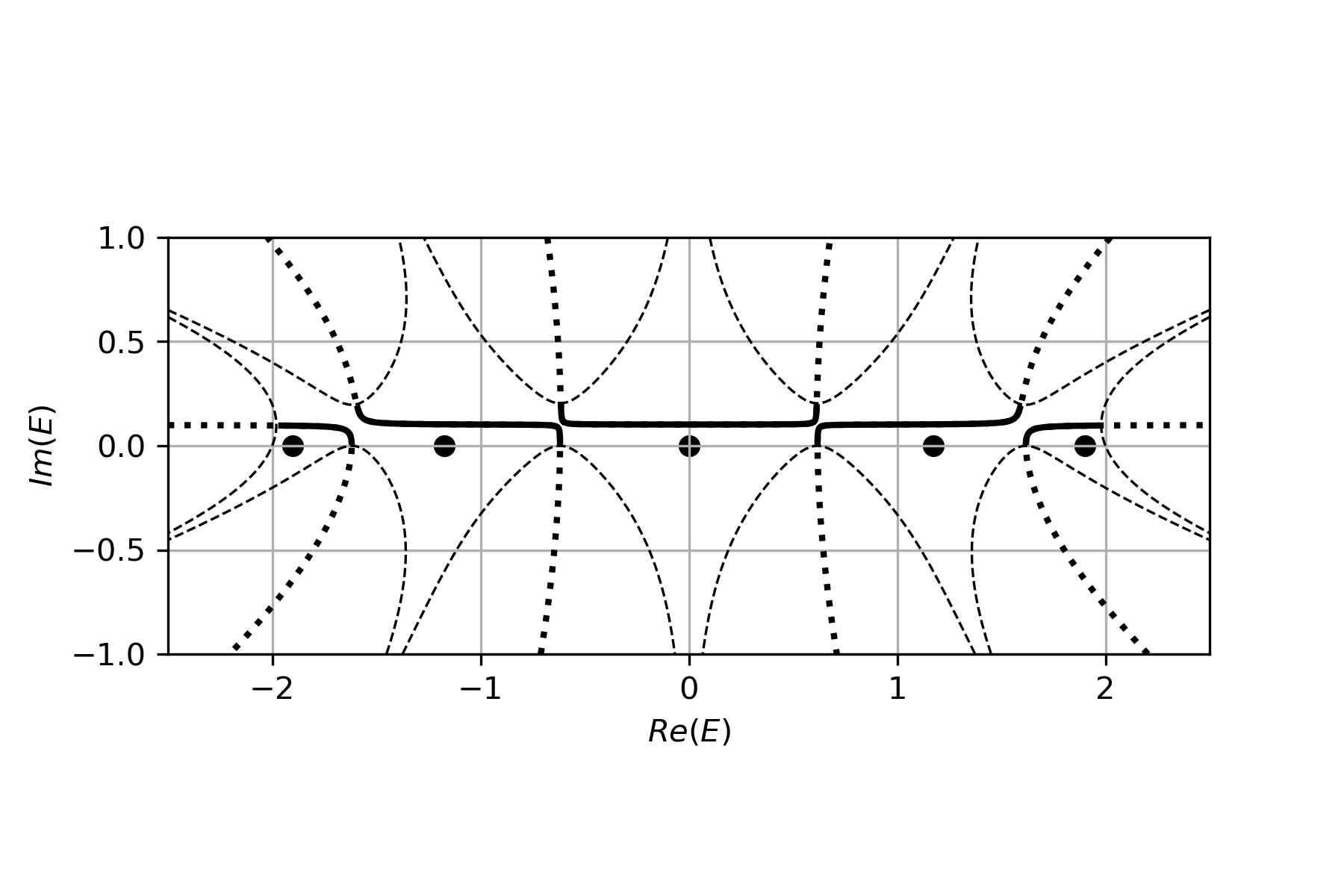}
 \caption{The spectrum for $L=5$ and $v=i/2$}
 \label{fig:L5a05}
\end{figure}
\begin{example}\label{ex:L5a05}
Figure \ref{fig:L5a05} shows the spectrum for $L=5$, $v=i/2$. 
In Figure \ref{fig:L5a05}, the dashed lines are the curves represented by $\Re\Delta_5(E) = \Re(2T_5(E/2)) + \frac{1}{2}\Im U_4(E/2)=\pm 2$ and the dotted lines are the curves represented by $\Im\Delta_5(E) = \Im(2T_5(E/2)) - \frac{1}{2}\Re U_4(E/2)=0$, where 
\begin{eqnarray*}
\Re(2T_5(E/2)) &=& x^5 - 10x^3y^2 - 5x^3 + 5xy^4 + 15xy^2 + 5x \\ 
\Im(2T_5(E/2)) &=& 5x^4y - 10x^2y^3 - 15x^2y + y^5 + 5y^3 + 5y \\
\Re U_4(E/2) &=& x^4 - 6x^2y^2 - 3x^2 + y^4 + 3y^2 + 1 \\
\Im U_4(E/2) &=& 4x^3y - 4xy^3 - 6xy 
\end{eqnarray*}
for $E=x+iy~((x, y)\in\mathbb{R}^2)$.
The solid lines are the spectrum (spectral arcs), and the filled circles on the real axis are $\alpha_1, \alpha_2, \ldots,\alpha_5$ in order from right to left. 
The figure shows that the spectral arcs are almost parallel to the real axis in the neighborhood of each $\alpha_j+v/L=\alpha_j + i/10(j=1, 2, \ldots, 5)$.
\end{example}

\medskip
We next consider the case of large $|v|$. 
\begin{thm}\label{thm:localstructure2}
Let $\beta_k$ be defined in Lemma \ref{lem:Taylor}, and $\beta_k(v, \kappa)$ be the element nearest $\beta_k$ of the Floquet spectrum $\sigma_{\kappa}(H)$.
For $v$ with sufficiently large $|v|$ and $\kappa\in [0, \pi]$ with sufficiently small $|\cos\kappa/v|$, 
$$
\beta_{k}(v, t) = \beta_k - \frac{2\sin^2\frac{k\pi}{L}}{L}\frac{1}{v} -(-1)^k\frac{2\sin^2\frac{k\pi}{L}}{L}\frac{\cos\kappa}{v} + O((|1/v|+|\cos\kappa/v|)^2),  
$$
where $\beta_k$ is defined in Lemma \ref{lem:Taylor}.
\end{thm}
\begin{proof} The equation $\Delta_L(E) - 2\cos\kappa = 0$ can be rewritten as
\begin{equation}
U_{L-1}(E/2) - \frac{2}{v}T_L(E/2) + \frac{2\cos\kappa}{v} = 0. 
\end{equation}
We recall 
$$
T_L(\beta_k/2)=(-1)^k, \quad U_{L-1}(\beta_k/2)=-\frac{L(-1)^k}{\sin^2\frac{k\pi}{L}},
$$
and Lemma \ref{lem:Taylorapprox}, hence we have
\begin{eqnarray*}
\beta_k(v, \kappa) &=& \beta_k + \frac{2T_{L}(\beta_k/2)}{U_{L-1}'(\beta_k/2)}\frac{1}{v} - \frac{2\cos\kappa}{U_{L-1}'(\beta_k/2)}\frac{1}{v} + O((|1/v|+|\cos\kappa/v|)^2) \\
&=&  \beta_k - \frac{2\sin^2\frac{k\pi}{L}}{L}\frac{1}{v} -(-1)^k\frac{2\sin^2\frac{k\pi}{L}}{L}\frac{\cos\kappa}{v} + O((|1/v|+|\cos\kappa/v|)^2). 
\end{eqnarray*}
\end{proof}

\begin{figure}[htbp]
 \centering
  \includegraphics[width=100mm]{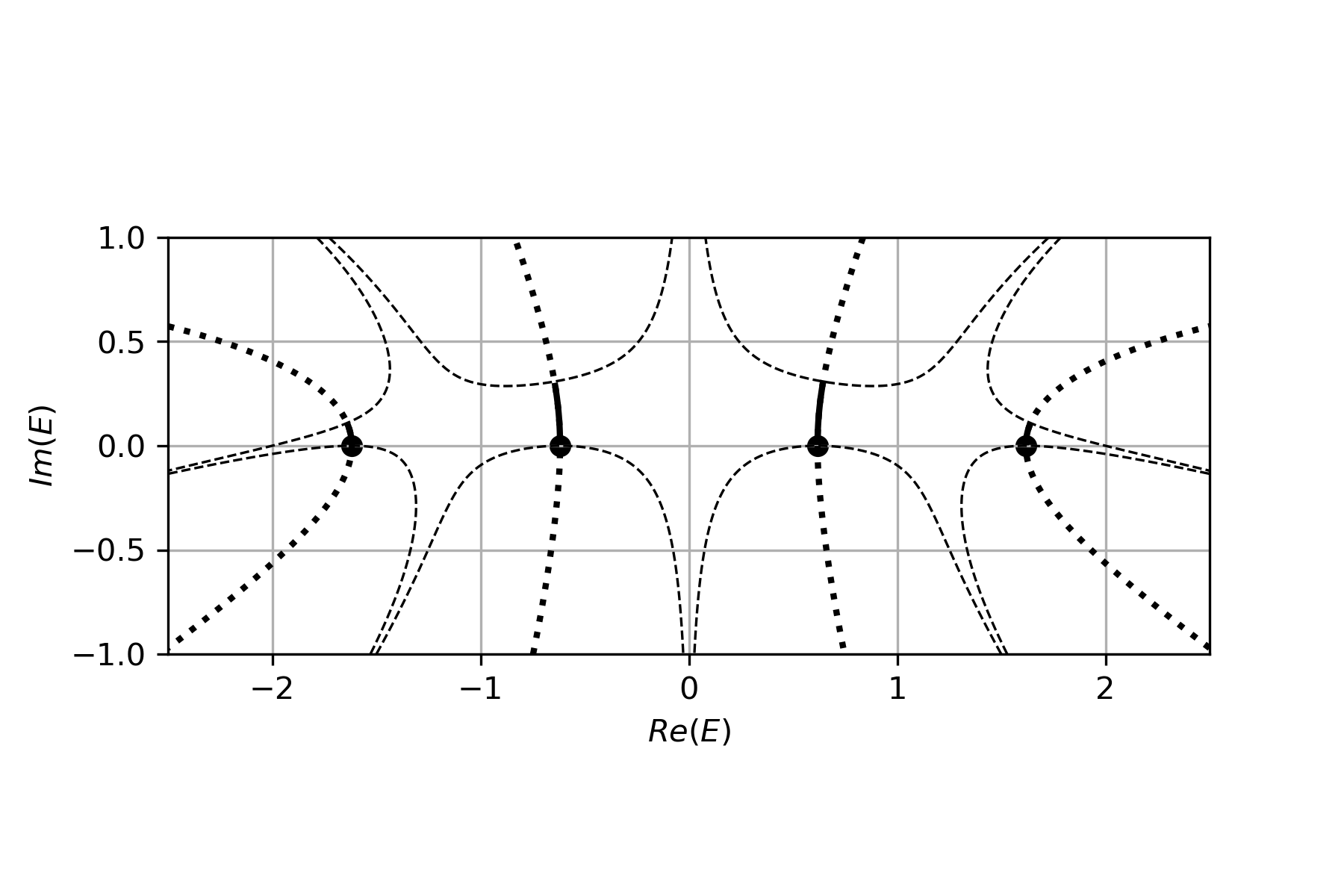}
 \caption{The spectrum for $L=5$ and $v=5i$}
 \label{fig:L5a5}
\end{figure}

\begin{example}
Figure \ref{fig:L5a5} shows the spectrum for $L=5$, $v=5i$. 
The discriminant is the same as in Example \ref{ex:L5a05} except for the value of $v$.
In Figure \ref{fig:L5a5}, the dashed lines are the curves represented by $\Re\Delta_5(E) = \Re(2T_5(E/2)) + 5\Im U_4(E/2)=\pm 2$ 
and the dotted lines are the curves represented by $\Im\Delta_5(E) = \Im(2T_5(E/2)) - 5\Re U_4(E/2)=0$ for $E=x+iy~((x, y)\in\mathbb{R}^2)$. 
The solid lines are the spectral arcs, and the filled circles on the real axis are $\beta_1, \beta_2, \beta_3$, and $\beta_4$ in order from right to left. 
Four curves extend from $\beta_1$, $\beta_2$, $\beta_3$, and $\beta_4$ almost in the direction of the imaginary axis, i.e., in the direction of $-1/v$. 
In the case of small $v$, five connected components of the spectrum exist, as in Figure \ref{fig:L5a05}, but as $v$ increases, there are four.
\end{example}

\medskip
\noindent
{\small\bf Acknowledgements} \\ 
The author would like to thank Dr. Yu Morishima for helpful suggestions on drawing figures in Python.

\begin{flushright}
Department of Information Technology \\
Faculty of Engineering \\
Tohoku Gakuin University \\
1-13-1 Chuo, Tagajo, Miyagi 985-0873, Japan 
\end{flushright}

\end{document}